\documentclass[twoside,11pt]{article}
\usepackage{amssymb,amsmath,amsthm}           %For double arrows
\usepackage{hyperref}
\usepackage{eucal}
\usepackage{a4wide}
\usepackage[svgnames]{xcolor}
\usepackage{mathptmx} %times font
\usepackage{graphicx}
\usepackage{tikz-cd}
 \usepackage{subfigure}
 \usepackage{fancyhdr}

\title{Generalisation of Chaplygin's Reducing Multiplier  Theorem with an application to multi-dimensional nonholonomic
dynamics\footnote{This research was made possible by a Georg Forster Experienced Researcher Fellowship  from the 
Alexander von Humboldt Foundation that funded a research stay of the author at TU Berlin.}}

\author{ Luis C.~Garc\'ia-Naranjo}

\numberwithin{equation}{section}
\numberwithin{table}{section}
\numberwithin{figure}{section}

\hypersetup{colorlinks=true,linkcolor=violet,citecolor=purple}

\newtheorem{theorem}{Theorem}[section]
\newtheorem{lemma}[theorem]{Lemma}
\newtheorem{proposition}[theorem]{Proposition}
\newtheorem{corollary}[theorem]{Corollary}

{\theoremstyle{definition}

\newtheorem{remark}[theorem]{Remark}
\newtheorem*{remarks*}{Remarks}

}

\newtheorem{innercustomgeneric}{\customgenericname}
\providecommand{\customgenericname}{}
\newcommand{\newcustomtheorem}[2]{%
  \newenvironment{#1}[1]
  {%
   \renewcommand\customgenericname{#2}%
   \renewcommand\theinnercustomgeneric{##1}%
   \innercustomgeneric
  }
  {\endinnercustomgeneric}
}

\newcustomtheorem{customhyp}{}

\newcommand{\defn}[1]{{\bfseries\itshape{#1}}}

\def\headcolour{\color{Grey}}
\def\restr#1{\,\vrule height1.2ex width.4pt
  depth0.8ex\lower0.4ex\hbox{\scriptsize $\,#1$}}

\newcommand{\R}{\mathbb{R}}

\newcommand{\I}{\mathbb{I}}
\newcommand{\J}{\mathbb{J}}
\newcommand{\g}{\mathfrak{g}}
\newcommand{\so}{\mathfrak{so}}

\newcommand{\SO}{\mathrm{SO}}

\newcommand{\D}{\mathcal{D}}

\newcommand{\tr}{\mathop\mathrm{tr}\nolimits}

\begin{document}

\maketitle

\begin{abstract}
A generalisation of Chaplygin's Reducing Multiplier  Theorem is given by providing  sufficient conditions for the
Hamiltonisation of Chaplygin nonholonomic
systems with an arbitrary number $r$ of degrees of freedom via Chaplygin's multiplier method. 
The crucial point in the construction is to  add an  hypothesis  
of geometric nature that controls the interplay between the kinetic energy metric and the non-integrability of the constraint distribution.
Such hypothesis can  be systematically examined
in concrete examples, and  is 
automatically satisfied in the case $r= 2$ encountered in the original formulation of Chaplygin's theorem. Our results are applied to prove
 the Hamiltonisation of 
a multi-dimensional generalisation of the problem of a symmetric rigid body with 
a flat face that rolls without slipping or spinning over a sphere.

\end{abstract}

{\small
\tableofcontents
}

\section{Introduction}

A substantial amount of research in nonholonomic mechanics in recent years has focused on Hamiltonisation (see e.g. \cite{BorMamChap,EhlersKoiller,FedJov, FassoRank2,BorMamHam, FedJov2, Jovan,LGN10,BolsBorMam2011, BalseiroGN,BolsBorMam2015}
and the references therein).
Roughly speaking this is the process by which, via symmetry reduction and a time reparametrisation, the equations of motion of certain  nonholonomic systems take a Hamiltonian form. 

The Hamiltonisation process has received special attention for the  the so-called Chaplygin systems, which are nonholonomic
systems with a specific type of symmetry~(defined in Section~\ref{S:Prelim}). The 
reduced equations  of motion for these systems
take the form of an unconstrained mechanical system subject to nonholonomic reaction forces of {\em gyroscopic} type.
According to \cite{FedJov},  it was Appel~\cite{Appel} who first proposed the idea of introducing a time reparametrisation to 
eliminate these forces, and the idea was taken up by Chaplygin who introduced the reducing multiplier method, and  proved his celebrated Reducing Multiplier Theorem \cite{ChapRedMult}.
The theorem states that if the reduced space has {\em two} degrees of freedom, the reduced equations may be written 
in Hamiltonian form after a time reparametrisation if and only if  there exists an invariant measure.
 This result has received enormous attention in the community of nonholonomic systems (see e.g.~\cite{FedJov,EhlersKoiller, BolsBorMam2015} and the references therein) and, 
despite the low-dimensional restriction on the dimension of the reduced space, it remains to be one of  the most solid theoretical results
 in the area of Hamiltonisation.

On the other hand, in recent years Fedorov and Jovanovi\'c have found a remarkable class of examples with arbitrary number
of degrees of freedom that allow a Hamiltonisation by  Chaplygin's method~\cite{FedJov, FedJov2, Jovan, JovaRubber, Jova18}.  
The treatment of  the authors in of all these examples  is independent of  previous theoretical efforts to generalise Chaplygin's Theorem (e.g. \cite{Iliev1985,Stanchenko,BorMamHam,Fernandez})
 and  the  underlying mechanism responsible for the Hamiltonisation of nonholonomic systems with arbitrary number of degrees of freedom
remained  a mystery. 

In this paper we present a generalisation of Chaplygin's Theorem that gives sufficient 
conditions for Hamiltonisation via Chaplygin's method for Chaplygin systems whose reduced space has an arbitrary number $r$ of degrees of
freedom. The crucial point  is to  add   hypothesis  \ref{hyp} (see section~\ref{S:theo})
which is of geometric nature and  controls the interplay between the kinetic energy metric and the non-integrability of the constraint distribution. This condition 
can be systematically analysed in concrete examples and is automatically satisfied in the case $r=2$ considered by Chaplygin. 

The  usefulness  of our generalisation is illustrated by explicitly applying it to prove the 
 Hamiltonisation of a concrete multi-dimensional nonholonomic system.
 Such system  consists of an $n$-dimensional  symmetric rigid body with 
a flat face that rolls without slipping or spinning over an $n-1$-dimensional sphere. This provides a new example of a Hamiltonisable nonholonomic Chaplygin  
system with arbitrary degrees of freedom.

\vspace{0.1cm}
\noindent {\bf Geometric interpretation and scope of the results}

After the first version of this paper was made available online, some works have appeared  that
further clarify the underlying  geometry and the relevance of the results in this paper.

The first is  Gaji\'c and Jovanovi\'c~\cite{Gajic}, that presents a novel application
of Chaplygin's reducing multiplier  beyond the nonholonomic setting.

The second is  Garc\'ia-Naranjo and Marrero~\cite{LGN-Marrero}, that continues the research started here 
taking an intrinsic geometric perspective, which shows that our main  result is a reformulation of 
Stanchenko~\cite[Proposition 2]{Stanchenko}
and Cantrijn et al.~\cite[Equation (18)]{CaCoLeMa}.
However, the  approach followed in the present paper, and continued in~\cite{LGN-Marrero},  seems to be more convenient to study
  concrete examples. In fact,~\cite{LGN-Marrero} also proves that the  
Hamiltonisation of the multi-dimensional Veselova problem established in~\cite{FedJov, FedJov2} may be 
explained in the light
of the results of this paper.

We finally mention Garc\'ia-Naranjo~\cite{LGN2019} that applies the results of this paper to establish the Hamiltonisation
of the multi-dimensional rubber Routh sphere.

\vspace{0.1cm}
\noindent {\bf Structure of the paper}

Section~\ref{S:Prelim} introduces the notation and recalls known results on Chaplygin systems.
The formulation and proof of our results is given in Section~\ref{S:theo}.  In section~\ref{S:Example} these are applied
 to prove the Hamiltonisation of our example.

\section{Preliminaries}
\label{S:Prelim}

A nonholonomic system consists of a triple $(Q,\D,L)$ where $Q$ is an  $n$-dimensional configuration manifold, 
 $\D\subset TQ$ is a  rank $r$ non-integrable distribution
on $Q$ that models $n-r$ independent linear constraints on the velocities,  and the Lagrangian $L:TQ\to \R$ is assumed to be of 
mechanical form, $L=K-U$,
 where the kinetic energy $K$ defines a Riemannian metric $\langle \cdot , \cdot \rangle$ on $Q$ and $U:Q\to \R$ is the potential energy.

This paper is concerned with nonholonomic $G$-\defn{Chaplygin systems}, or simply  a \defn{Chaplygin systems}, which are nonholonomic systems $(Q,\D,L)$ with the additional property that 
 there is a Lie group $G$ acting freely and properly on $Q$, satisfying the following properties:
\begin{enumerate}
\item $G$ acts by isometries on $Q$ and the potential energy $U$ is $G$-invariant,
\item $\D$ is $G$-invariant in the sense that $T_q\Phi_g(\D_q)=\D_{\Phi_g( q)}$, for all $g\in G$ where $\Phi_g: Q\to Q$ is the action
diffeomorphism defined by $g$,
\item for every $q\in Q$  the following direct sum splitting holds
\begin{equation}
\label{eq:chap-condition}
T_qQ=\D_q \oplus ( \g \cdot q),
\end{equation}
where $ \g$ denotes the Lie algebra of $G$ and $\g \cdot q$  the tangent space to the $G$-orbit through $q$.
\end{enumerate}

These systems often appear in applications and have been studied by several authors, e.g.
 \cite{Stanchenko,Koi, BKMM,  CaCoLeMa,EhlersKoiller,BM2015}.
The reduced configuration manifold $S:=Q/G$ is called the \defn{shape space}. Note that  
because of \eqref{eq:chap-condition} the dimension of $S$
coincides with the rank $r$ of $\D$, and the dimension of $G$ is $n-r$.

As first explained by Koiller \cite{Koi}, for  Chaplygin systems the constraint distribution 
$\D$ may be interpreted as the horizontal space of a principal connection on the
 principal $G$-bundle $\pi : Q\to S$, and the symmetry leads
to a reduced  system on the space 
$\D/G$ which is isomorphic to  $TS$. The reduced equations on $TS$ take the form 
of an unconstrained mechanical system
on $S$ subject to a {\em gyroscopic} force:
\begin{equation}
\label{P:LagTS}
\frac{d}{dt}\left ( \frac{\partial \ell}{\partial \dot s^i} \right )  - \frac{\partial \ell}{\partial s^i}= -\sum_{j,k=1}^rC^k_{ij}(s)\dot s^j \frac{\partial \ell}{\partial \dot s^k}, \qquad i=1,\dots, r.
\end{equation}
We now proceed to define the objects in the above equations. First, 
 $(s^1, \dots, s^r)$ are local coordinates on  $S$. Next,   
 $\ell:TS\to \R$ is the \defn{reduced Lagrangian} defined by
\begin{equation}
\label{eq:def-reduced-lag}
\ell(s,\dot s)= L\left (q,\mbox{hor}_q (\dot s) \right ), \qquad \mbox{where} \quad q\in \pi^{-1}(s),
\end{equation}
and $\mbox{hor}_q (\dot s)$ denotes the \defn{horizontal lift}  of $\dot s \in T_sS$ at $q$,  which is the tangent vector in $T_qQ$  characterised by 
the conditions that  $\mbox{hor}_q (\dot s)\in \D_q$ and $T_q\pi (\mbox{hor}_q (\dot s)) = \dot s$. The reduced Lagrangian  is 
locally  written as
\begin{equation}
\ell(s,\dot s)= \frac{1}{2}\sum_{i,j=1}^r K_{ij}(s)\dot s^i \dot s^j -U(s),
\end{equation}
where  $U:S\to \R$ now denotes the reduction of the $G$-invariant potential on $Q$, and  the coefficients $K_{ij}(s)$ are given by
\begin{equation*}
K_{ij}(s) = \left \langle  \mbox{  hor}_q \left ( \frac{\partial}{\partial s^i}  \right ) \,  ,  \,  \mbox{  hor}_q \left ( \frac{\partial}{\partial s^j}  \right )
\right \rangle_q, \qquad q\in \pi^{-1}(s),
\end{equation*}
and define a Riemannian metric on $S$. As it is standard, we will denote by $K^{ij}(s)$ the entries of the corresponding inverse
matrix, i.e.
\begin{equation*}
\sum_{k=1}^rK_{ik}(s)K^{kj}(s) = \delta_i^j, 
\end{equation*}
for all $i,j=1, \dots, r$, where here, and throughout, the symbol $\delta$ is reserved for the Kronecker delta.

Finally, the $s$ dependent coefficients $C_{ij}^k$, $i,j,k\in \{1,\dots, r\}$, in~\eqref{P:LagTS} are locally given by
\begin{equation}
\label{eq:Coeffs}
C_{ij}^k(s) = \sum_{l=1}^r K^{kl}(s) \left \langle  \left [ \, \mbox{ hor}_q \left ( \frac{\partial}{\partial s^i} \right ) \, , \,  \mbox{ hor}_q \left ( \frac{\partial}{\partial s^j} \right )\, \right   ] \, , \, 
    \mbox{ hor}_q \left ( \frac{\partial}{\partial s^l} \right ) \right \rangle_q, \qquad q\in \pi^{-1}(s),
\end{equation}
where $[\cdot, \cdot ]$ denotes the commutator of vector fields on $Q$.  Note that they are
 skew-symmetric on the lower indices, $C_{ij}^k=-C_{ji}^k$, and as a consequence the energy 
 $E=\sum_{i=1}^r \frac{\partial \ell}{\partial \dot s^i} \dot s^i-\ell$
is preserved. Inspired by this property we will refer to $C_{ij}^k$ as the \defn{  gyroscopic coefficients}.

The gyroscopic coefficients are central to this work. They are 
 the coordinate 
representation of a $(1,2)$ tensor field on $S$ that we call the 
\defn{gyroscopic tensor}
\begin{equation*}
\mathcal{T}= \sum_{i,j,k=1}^rC_{ij}^k(s) \, ds^i\otimes ds^j \otimes \frac{\partial}{\partial s^k}.
\end{equation*}
An intrinsic definition of this tensor, together with a geometric study of its
properties in relation to Chaplygin systems is given in Garc\'ia-Naranjo and Marrero~\cite{LGN-Marrero}.\footnote{The 
gyroscopic tensor actually coincides, up to a sign, with the tensor field $C$ considered by Koiller~\cite[Proposition 8.5]{Koi}
and Cantrijn et al~\cite[Page 337]{CaCoLeMa} (see~\cite{LGN-Marrero}).}

For the rest of the paper we will take a Hamiltonian approach. 
Define the momenta   $p_i :=\frac{\partial \ell}{\partial \dot s^i}$, 
so that   $(s^1, \dots , s^r, p_1, \dots p_r)$ are canonical coordinates for the cotangent bundle $T^*S$.  Denote by $H$ the reduced Hamiltonian:
\begin{equation}
\label{eq:Ham}
H:T^*S\to \R, \qquad H(s,p)=\sum_{i=1}^r p_i \dot s^i-\ell= \frac{1}{2}\sum_{i,j=1}^rK^{ij}(s)p_ip_j+U(s).
\end{equation}
Equations~\eqref{P:LagTS} may  be rewritten as the 
following first order system on  $T^*S$:\footnote{The form of the equations~\eqref{P:LagTS} and~\eqref{eq:reduced-Hamilt}  
indicates that there is an interesting connection
between the gyroscopic coefficients and the structure coefficients of other sophisticated geometric frameworks
 that have been developed
 to formulate the equations of motion of nonholonomic systems~\cite{Grab,Leon}.}
\begin{equation}
\label{eq:reduced-Hamilt}
\frac{ds^i}{dt} = \frac{\partial H}{\partial p_i}, \qquad  \frac{dp_i}{dt} =- \frac{\partial H}{\partial s^i}- \sum_{j,k=1}^rC_{ij}^kp_k  \frac{\partial H}{\partial p_j}, \qquad i=1,\dots, r.
\end{equation}

\subsection{Invariant measures for Chaplygin systems}

The existence of a smooth invariant measure for Eqns.~\eqref{eq:reduced-Hamilt}   is 
intimately related to the condition that a certain 1-form $\Theta$ on $S$ is exact  \cite{CaCoLeMa, FedGN2015}. A local expression for
  $\Theta$ may be given in terms of the gyroscopic coefficients by:
$$
\Theta :=  -\sum_{i,j=1}^r C_{ij}^j  \, ds^i.
$$

In the following theorem recall that   $T^*S$ is equipped with the \defn{ Liouville measure} $\nu$, that in local bundle coordinates is
given by $\nu =ds^1 \cdots ds^r \, dp_1 \cdots dp_r$. Recall also that a volume form $\mu$ on $T^*S$ is called a \defn{basic measure} if its density with respect to 
$\nu$ does not depend  on the momenta $p_j$.
\begin{theorem}[Cantrijn et al. \cite{CaCoLeMa}]
\label{T:measures-1-form}
Denote by $\nu$  the Liouville volume form on $T^*S$
and consider a   a mechanical Hamiltonian $H$ given as in~\eqref{eq:Ham}.  The  reduced  equations of motion 
\eqref{eq:reduced-Hamilt}   preserve the basic measure
\begin{equation}
\label{eq:measure}
\mu = \exp (\sigma) \, \nu, \qquad \sigma\in C^\infty(S),
\end{equation}
 if and only if $\Theta$ is exact with $\Theta =d\sigma$.
\end{theorem}
The proof of the main  result of this paper (Theorem~\ref{th:Chaplygin-Gen}
 in section~\ref{S:theo} below), uses the
fact that $\Theta=d\sigma$ is a {\em necessary} condition for the invariance of the measure $\mu$ given by~\eqref{eq:measure}.

\begin{proof}
 In local bundle coordinates, 
$$\mu = \exp(\sigma (s)) \, ds^1 \cdots ds^r \, dp_1 \cdots dp_r.$$
In view of \eqref{eq:reduced-Hamilt}, the condition for $\mu$ to be invariant is that 
\begin{equation*}
\begin{split}
0&=\sum_{i=1}^r \left [  \frac{\partial}{\partial s^i} \left (  \exp(\sigma (s)) \frac{\partial H}{\partial p_i} \right )
-  \frac{\partial}{\partial p_i}\left (  \exp(\sigma (s))  \left ( \frac{\partial H}{\partial q^i}+\sum_{j,k=1}^r C_{ij}^kp_k  \frac{\partial H}{\partial p_j} \right )  \right )
\right ] \\
&=  \exp(\sigma (s))\sum_{i=1}^r  \left (  \frac{\partial \sigma }{\partial s^i} +\sum_{j=1}^r C_{ij}^j \right )  \frac{\partial H}{\partial p_i},   
\end{split}
\end{equation*}
where we have used the skew-symmetry on the lower indices of the coefficients $C_{ij}^k$ to cancel the terms involving second derivatives with respect to the momenta.
Since the above equality should  hold for arbitrary $(s^1, \dots , s^r, p_1, \dots p_r)$, 
and \eqref{eq:Ham} implies $ \frac{\partial H}{\partial p_i}=\sum_{j=1}^rK^{ij}(s)p_j$ where   $K^{ij}(s)$ are
the coefficients of an invertible matrix,   then necessarily
\begin{equation}
\label{eq:Inv-Measure}
 \frac{\partial \sigma }{\partial s^i} = -\sum_{j=1}^r C_{ij}^j , \qquad i=1, \dots, r,
\end{equation}
and   $\Theta =d\sigma$. The sufficiency of this condition follows immediately from the above analysis.
\end{proof}
  
\subsection{Chaplygin's Reducing Multiplier Method and Theorem}

Chaplygin's reducing multiplier method attempts to find a smooth function  $\phi:S\to \R$ such that,  after the time and momentum reparametrisation 
\begin{equation*}
dt = \exp( - \phi (s))\, d\tau, \qquad  p_i= \exp(- \phi (s))\,\tilde p_i, \qquad i=1,\dots, r,
\end{equation*}
 the equations of motion~\eqref{eq:reduced-Hamilt} transform into
Hamiltonian form:
\begin{equation}
\label{eq:HamFormResc}
\frac{ds^i}{d\tau} = \frac{\partial   \tilde H}{\partial \tilde p_i}, \qquad  \frac{d \tilde p_i}{d\tau} =- \frac{\partial \tilde H}{\partial s^i}, \qquad i=1,\dots, r,
\end{equation}
where $\tilde H(s_i, \tilde p_i)=H(s_i, \exp(- \phi (s))\,\tilde p_i)$. The process described above is often termed \defn{Chaplygin Hamiltonisation}. The contribution of this paper is to give sufficient conditions for the existence of $\phi$ that
can be systematically examined in concrete examples. Here we recall two well-established results. The first one is that 
the existence of a basic, smooth invariant measure is a necessary condition for Chaplygin Hamiltonisation (see 
 e.g.  \cite{FedJov} and \cite{EhlersKoiller}).

\begin{proposition}
\label{P:measure}
Suppose that a nonholonomic Chaplygin system allows a 
  Chaplygin Hamiltonisation by  the time and momentum reparametrisation 
\begin{equation*}
dt = \exp( - \phi (s))\, d\tau, \qquad  p_i= \exp(- \phi (s))\,\tilde p_i, \qquad i=1,\dots, r.
\end{equation*}
Then, its reduced equations of motion~\eqref{eq:reduced-Hamilt} possess the invariant measure $\mu =\exp(\sigma (s)) \, \nu$, where
$\nu$ is the Liouville volume form on $T^*S$ and $\sigma=(r-1)\phi$.
\end{proposition}
\begin{proof}
By Liouville's Theorem, the transformed equations~\eqref{eq:HamFormResc} preserve the measure
\begin{equation*}
\begin{split}
\tilde \mu &= ds^1\wedge \cdots  \wedge ds^r  \wedge  d\tilde p_1 \wedge   \cdots   \wedge d\tilde p_r \\
&= \exp( r \phi (s)) \, ds^1 \wedge \cdots  \wedge ds^r  \wedge d p_1\wedge  \cdots  \wedge d p_r. 
\end{split}
\end{equation*}
Therefore, the equations in the original time variable $t$ preserve the measure $\mu=  \exp( - \phi (s))  \tilde \mu$.
\end{proof}

The celebrated  Chaplygin's  Reducing Multiplier Theorem establishes that if the dimension of the shape space  is 2, the 
existence of the invariant measure $\mu =\exp(\sigma (s)) \, \nu$ is not only necessary, but also sufficient for Chaplygin Hamiltonisation.
More precisely:

\begin{theorem}[Chaplygin's Reducing Multiplier Theorem \cite{ChapRedMult}]
\label{th:Chaplygin}
Suppose that $r=2$ and that the reduced equations  \eqref{eq:reduced-Hamilt} possess the invariant measure 
\begin{equation*}
\mu =\exp(\sigma (s)) \, \nu,  \qquad \sigma\in C^\infty(S),
\end{equation*}
where $\nu$ is the Liouville volume form on $T^*S$.
Then after the time and momentum reparametrisation 
\begin{equation*}
dt = \exp( - \sigma (s))\, d\tau, \qquad  p_i= \exp(- \sigma (s))\,\tilde p_i, \qquad i=1,2,
\end{equation*}
the equations \eqref{eq:reduced-Hamilt} transform to Hamiltonian form
\begin{equation*}
\frac{ds^i}{d\tau}= \frac{\partial \tilde H}{\partial \tilde p_i}, \qquad \frac{d \tilde p_i}{d\tau}= -\frac{\partial \tilde H}{\partial  s_i}, \qquad i=1,2,
\end{equation*}
where $\tilde H (s^i,\tilde p_i) =H(s^i,\exp( - \sigma (s))\tilde p_i)$.
\end{theorem}

The theorem of Chaplygin is a particular instance of  our main Theorem~\ref{th:Chaplygin-Gen} given below.

\section{Generalisation of Chaplygin's Reducing Multiplier Theorem}
\label{S:theo} 

The generalisation of Chaplygin's Theorem~\ref{th:Chaplygin} that we present gives sufficient conditions for the Hamiltonisation
of Chaplygin systems whose shape space $S$ has dimension $r\geq 2$. We replace the low dimensional assumption on $S$ of Chaplygin's Theorem, by
the following hypothesis on the gyroscopic coefficients $C_{ij}^k$:

\begin{customhyp}{(H)}
\label{hyp}
The gyroscopic coefficients $C_{ij}^k$  satisfy
\begin{equation*}
C_{ij}^k= 0, \quad \mbox{for} \quad k\neq i\neq j\neq k, \qquad C_{ij}^j=C_{ik}^k \quad \mbox{for all } \quad j, k\neq i.
\end{equation*}
\end{customhyp}
It may seem that condition \ref{hyp}, as formulated here, depends on the choice of coordinates. Lemma~\ref{L:Intrinsic}  
below shows that 
this is not the case.

\begin{remark}
\label{2-dim}
It is immediate to check that \ref{hyp} is satisfied automatically if $r=2$. 
\end{remark}

\begin{remark} In the proofs  of Lemma~\ref{L:Intrinsic} and Theorem~\ref{th:Chaplygin-Gen} below, we will find  it  useful to work with the following, equivalent formulation of condition~\ref{hyp}:
\begin{equation}
\label{eq:H-alt}
C_{ij}^k= \chi_i \delta_j^k - \chi_j\delta_i^k,
\end{equation}
for certain locally defined functions $\chi_i$ on $S$, $i=1,\dots r$. 
\end{remark}

\begin{lemma}
\label{L:Intrinsic}
Let $(s^1,\dots, s^r)$ and $(y^1, \dots, y^r)$ be coordinates on a  neighbourhood of $S$, and  let $C_{ij}^k$ and 
$\tilde C_{\alpha \beta }^\gamma$
 be the corresponding gyroscopic coefficients.
 If  $C_{ij}^k$ satisfy {\em \ref{hyp}}, then the same is true about $\tilde C_{\alpha \beta }^\gamma$.
\end{lemma}
\begin{proof}
Due to the coordinate transformation rules for  the gyroscopic tensor $\mathcal{T}$ we have
\begin{equation*}
\label{eq:aux-lemma0}
\tilde C_{\alpha \beta}^\gamma =\sum_{i,j,k=1}^r C_{ij}^k
   \frac{\partial s^i}{\partial y^\alpha}  \frac{\partial s^j}{\partial y^\beta}
\frac{\partial y^\gamma}{\partial s^k},\quad \mbox{for all} \quad \alpha, \beta, \gamma =1, \dots, r.
\end{equation*}
Hence, \eqref{eq:H-alt} implies
\begin{equation*}
\begin{split}
\tilde C_{\alpha \beta}^\gamma &=\sum_{i,j=1}^r\left (\chi_i\frac{\partial s^i}{\partial y^\alpha}  \frac{\partial s^j}{\partial y^\beta}
\frac{\partial y^\gamma}{\partial s^j} \right )-\sum_{i,j=1}^r\left (\chi_j \frac{\partial s^i}{\partial y^\alpha}  \frac{\partial s^j}{\partial y^\beta}
\frac{\partial y^\gamma}{\partial s^i} \right ) \\
&=\left ( \sum_{i=1}^r \chi_i\frac{\partial s^i}{\partial y^\alpha} \right ) \delta^\gamma_\beta - \left ( \sum_{j=1}^r \chi_j\frac{\partial s^j}{\partial y^\beta}
\right ) \delta^\gamma_\alpha.
\end{split}
\end{equation*}
Therefore, the coefficients $\tilde C_{\alpha \beta}^\gamma$ also satisfy \eqref{eq:H-alt} with corresponding $\tilde \chi_\alpha
 = \sum_{i=1}^r \chi_i\frac{\partial s^i}{\partial y^\alpha} $.
\end{proof}

The  lemma shows that~\ref{hyp} contains  intrinsic geometric information about the interplay 
between the kinetic energy metric and the
constraint distribution (see Remark~\ref{rmk:intrinsic} below for more details).

Recall from Proposition~\ref{P:measure} that the existence of a basic invariant measure 
is a  necessary condition for Chaplygin Hamiltonisation. Moreover, its density with respect to the Liouville volume determines 
the corresponding time and momentum rescaling. Our main result, stated in the 
theorem below, shows that, in the presence of a  basic invariant measure, the Chaplygin Hamiltonisation of the system is guaranteed  by condition~\ref{hyp}.

\begin{theorem}[Generalisation of Chaplygin's Reducing Multiplier Theorem]
\label{th:Chaplygin-Gen}
Suppose that $r\geq 2$ and that the reduced equations  \eqref{eq:reduced-Hamilt} possess the invariant measure
\begin{equation*}
\mu =\exp(\sigma (s)) \, \nu,  \qquad \sigma\in C^\infty(S),
\end{equation*}
where  $\nu$ is the Liouville volume form on $T^*S$.
Suppose moreover that the gyroscopic coefficients satisfy {\em \ref{hyp}} everywhere on $S$.
Then, after the time and momentum reparametrisation 
\begin{equation}
\label{eq:time-mom-repar}
dt = \exp \left ( \frac{ \sigma (s)}{1-r} \right )\, d\tau, \qquad  p_i= \exp\left ( \frac{ \sigma (s)}{1-r} \right )\,\tilde p_i, \qquad i=1,\dots, r,
\end{equation}
the equations \eqref{eq:reduced-Hamilt} transform to Hamiltonian form
\begin{equation*}
\frac{ds^i}{d\tau}= \frac{\partial \tilde H}{\partial \tilde p_i}, \qquad \frac{d \tilde p_i}{d\tau}= -\frac{\partial \tilde H}{\partial  s_i}, \qquad i=1,\dots ,r,
\end{equation*}
where $\tilde H (s^i,\tilde p_i) =H \left (s^i, \exp \left ( \frac{ \sigma (s)}{1-r} \right )\,\tilde p_i \right )$.
\end{theorem}
Since \ref{hyp} holds automatically when $r=2$ (Remark~\ref{2-dim}), this is indeed a generalisation of Chaplygin's Theorem~\ref{th:Chaplygin}.

\begin{proof}
By the chain rule we have 
\begin{equation*}
\frac{\partial \tilde H}{\partial s^i} =  \frac{\partial H}{\partial s^i}+\frac{1}{1-r}\sum_{j=1}^r \frac{\partial H}{\partial p_j} \tilde p_j \exp \left (  \frac{ \sigma (s)}{1-r}
\right )
\frac{\partial \sigma }{\partial s^i}, \qquad 
\frac{\partial \tilde H}{\partial \tilde p_i}=\exp \left (  \frac{ \sigma (s)}{1-r}
\right )
 \frac{\partial  H}{\partial  p_i}, 
\end{equation*}
and hence
\begin{equation}
\label{eq:proof-gen-aux0}
\frac{\partial  H}{\partial s^i} =  \frac{\partial \tilde  H}{\partial s^i} - \frac{1}{1-r} \sum_{j=1}^r \frac{\partial \tilde H}{\partial \tilde p_j} \tilde p_j 
\frac{\partial \sigma }{\partial s^i}, \qquad 
\frac{\partial  H}{\partial  p_i}= \exp \left (  \frac{ -\sigma (s)}{1-r}\right )
 \frac{\partial  \tilde H}{\partial \tilde  p_i}, \qquad i=1,\dots, r.
\end{equation}
Therefore, in view of the first equation in   \eqref{eq:reduced-Hamilt} we obtain
\begin{equation}
\label{eq:proof-r=2-aux1}
\frac{ds^i}{d\tau}= \exp \left (  \frac{ \sigma (s)}{1-r}
\right )
 \frac{ds^i}{dt} =\frac{\partial \tilde H}{\partial \tilde p_i}, \qquad i=1,\dots, r.
\end{equation}

On the other hand we have
\begin{equation*}
\begin{split}
\frac{d\tilde p_i}{d\tau} & =\exp \left (  \frac{ \sigma (s)}{1-r}\right ) \frac{d \tilde p_i}{dt}=  \frac{d  p_i}{dt} -\frac{p_i}{1-r}\sum_{j=1}^r\frac{\partial \sigma }{\partial s_j}\frac{ds^j}{dt}. 
\end{split}
\end{equation*}
Using now both equations in \eqref{eq:reduced-Hamilt}, the above equation becomes
\begin{equation}
\frac{d\tilde p_i}{d\tau} 
= - \frac{\partial H}{\partial s^i}- \sum_{j,k=1}^r C_{ij}^kp_k \frac{\partial H}{\partial p_j}
-\frac{p_i}{1-r}\sum_{j=1}^r\frac{\partial \sigma }{\partial s_j}\frac{\partial H}{\partial p_j} ,
\end{equation}
which in view of \eqref{eq:proof-gen-aux0} gives
\begin{equation}
\label{eq:general-after-chain-rule}
\frac{d\tilde p_i}{d\tau}   =-\frac{\partial  \tilde H}{\partial s^i} + \sum_{j=1}^r \left  [ \frac{1}{ 1-r} \left ( \tilde p_j \frac{\partial \sigma }{\partial s^i}- \tilde p_i \frac{\partial \sigma }{\partial s^j} \right )
-\sum_{k=1}^rC_{ij}^k \tilde p_k \right ] \frac{\partial \tilde H}{\partial \tilde  p_j}, \qquad i=1,\dots ,r.
\end{equation}
Using that  \ref{hyp} holds, we use~\eqref{eq:H-alt} to simplify
\begin{equation}
\label{eq:proof-aux1}
\sum_{k=1}^rC_{ij}^k \tilde p_k =  \tilde p_j   \chi_i - \tilde  p_i \chi_j.
\end{equation}
On the other hand, the assumption that the measure  $\mu$ is preserved by the flow, implies, by  Theorem \ref{T:measures-1-form},
 that \eqref{eq:Inv-Measure} holds.
Combining these equations with \ref{hyp} formulated as~\eqref{eq:H-alt}, we have
\begin{equation}
\label{eq:proof-aux2}
\begin{split}
   \tilde p_j \frac{\partial \sigma }{\partial s^i}- \tilde p_i \frac{\partial \sigma }{\partial s^j} &=
-\tilde p_j \sum_{k=1}^rC_{ik}^k + \tilde p_i \sum_{k=1}^rC_{jk}^k =
-\tilde p_j \sum_{k=1, k\neq i }^r\chi_i  + \tilde p_i \sum_{k=1, k\neq j}^r\chi_j \\
&= (r-1)\left ( \tilde p_i \chi_j - \tilde p_j \chi_i \right ).
\end{split}
\end{equation}
Substitution of \eqref{eq:proof-aux1} and \eqref{eq:proof-aux2} into \eqref{eq:general-after-chain-rule}
 leads to $\frac{d\tilde p_i}{d\tau}   =-\frac{\partial  \tilde H}{\partial s^i}$, for
$ i=1,\dots ,r$, as required.
\end{proof}

\begin{remark}
\label{rmk:intrinsic}
The 
 notion of {\em $\phi$-simple Chaplygin systems} introduced in~\cite{LGN-Marrero}  gives a  
 coordinate free characterisation of the  systems that satisfy the hypothesis of Theorem~\ref{th:Chaplygin-Gen}.
 Such concept is inspired by the observation that \ref{hyp}  is equivalent to the the existence of a 1-form $\chi$ on 
$S$ such that\footnote{This observation
is made in~\cite{LGN-Marrero} and was also  indicated by one of the referees of the paper.}   
 $$\mathcal{T}(Y,Z)=\chi(Y)Z-\chi(Z)Y,$$ for any two vector fields $Y, Z$ on $S$.
\end{remark}
\begin{remark}
The reduced equations of motion~\eqref{eq:reduced-Hamilt} 
 may be formulated in almost Hamiltonian form with respect to a 2-form on $T^*S$ that
 is non-degenerate but fails to be closed~\cite{Stanchenko,EhlersKoiller}.
 An equivalent geometric formulation of Theorem~\ref{th:Chaplygin-Gen} states that such 2-form
 is closed after multiplication by the conformal factor $f(s)= \exp \left ( \frac{ \sigma (s)}{r-1} \right )$, see~\cite{LGN-Marrero}.
 \end{remark}
 
Our result on Chaplygin Hamiltonisation persists under the addition of an invariant potential. More precisely we have:
\begin{corollary}
\label{C:potential}
Suppose that a $G$-Chaplygin system satisfies the hypothesis of Theorem~\ref{th:Chaplygin-Gen}. Then, after the addition
of a $G$-invariant potential energy, the system continues to satisfy the hypothesis of Theorem~\ref{th:Chaplygin-Gen} so it  
allows a Chaplygin Hamiltonisation by the same reparametrisation of time and momenta.
\end{corollary}
\begin{proof}
It is immediate to see from~\eqref{eq:reduced-Hamilt} that the system with the extra potential preserves the same volume form. 
On the other hand,
the system with the extra potential also satisfies condition~\ref{hyp} 
since the definition of the gyroscopic coefficients is independent of the potential.
\end{proof}

\section{Example: Hamiltonisation of a multi-dimensional symmetric rigid body  with a flat face that rubber rolls on the outer surface of a sphere}
\label{S:Example}

Following Borisov et al~\cite{BorMamRubber},  consider the motion of a rigid body with a flat face which is at every time  tangent to a sphere of
radius $R$ that is fixed in space, see Figure~\ref{F:fig1}. The body is subject to a {\em rolling} nonholonomic constraint 
that prevents slipping of its  flat face  over the surface of sphere, and to a {\em rubber}\footnote{This terminology was introduced by Ehlers, Koiller
and coauthors in~\cite{EhlersKoiller} and \cite{KoillerRubber} and is widely used in the literature.} nonholonomic constraint
that prohibits {\em spinning}, namely  it forbids  rotations of the body about the normal vector to the sphere at the contact point $P$. 
This problem without the rubber constraint
was first considered by Woronets~\cite{Vor1,Vor2}.

We shall give a multi-dimensional generalisation of the system and, assuming some symmetry of the body, 
 prove its Hamiltonisation by applying Theorem~\ref{th:Chaplygin-Gen}.

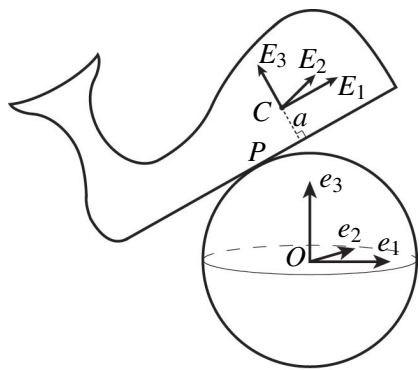
\begin{figure}[h]
\centering
\includegraphics[totalheight=5cm]{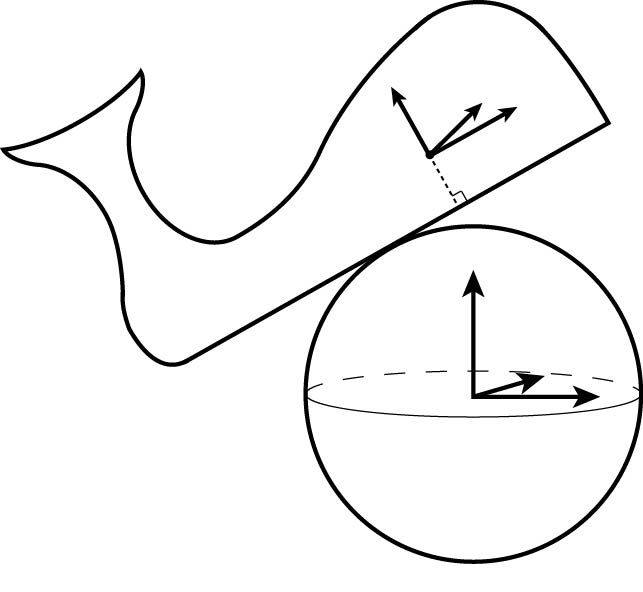}
 \put (-15,51) {$e_1$} \put (-30,56) {$e_2$} \put (-37,75) {$e_3$} \put (-50,45) {$O$}  \put (-62,100) {$C$}  \put (-47,98) {$a$} \put (-30,110) {$E_1$} \put (-45,120) {$E_2$} 
  \put (-60,122) {$E_3$}  \put (-64,85) {$P$} 
\caption{Body with a flat face that rubber rolls over a fixed sphere.}
\label{F:fig1}
\end{figure}

\subsection{The 3-dimensional case}
Consider a moving {\em body frame} $\{E_1, E_2, E_3\}$  that is rigidly attached to the body at its centre of mass  $C$ 
and is such that the $E_3$ axis is parallel to the outward normal vector to the
sphere at the contact point $P$ (see Figure~\ref{F:fig1}). Consider also a fixed {\em space frame}  $\{e_1, e_2, e_3\}$
which is attached to the centre $O$ of the sphere. As usual, the change of basis matrix between the two coordinate systems is an
element $g\in \SO(3)$ that determines the attitude of the body.

Let $x=(x_1,x_2,x_3)^T\in \R^3$ be the space coordinates of the vector $\overrightarrow{OC}$. Its corresponding body coordinates are
\begin{equation}
X=(X_1,X_2,X_3)^T=g^{-1}x.
\end{equation}
The condition that the flat face of the body is always tangent to the sphere gives the holonomic constraint
\begin{equation}
\label{eq:hol-const}
X_3=R+a,
\end{equation}
where $a$ is the distance between $C$ and the flat face, positively measured in the direction of $E_3$.\footnote{Note that $a=-z$ in the notation of~\cite{BorMamRubber}.}

The configuration of the system is completely determined by  $((X_1,X_2),g)\in \R^2\times \SO(3)$ so the configuration space
is $Q=\R^2\times \SO(3)$.
 In particular, the space coordinates
of the vector $\overrightarrow{OP}$ are the entries of the vector $RgE_3$. 

Let  $\Omega =(\Omega_1, \Omega_2, \Omega_3)^T
 \in \R^3$ be the angular velocity
of the body written in the body frame. Then
\begin{equation}
g^{-1}\dot g = \begin{pmatrix} 0 & -\Omega_3 & \Omega_2 \\ \Omega_3 & 0 & -\Omega_1 \\
-\Omega_2 & \Omega_1 & 0 \end{pmatrix} \in \so(3).
\end{equation}
The constraint that the body rolls without slipping over the sphere is 
\begin{equation}
\label{eq:Rolling-const-3D}
\dot X=  -R \Omega \times E_3,
\end{equation}
where $\times$ denotes the usual cross product in $\R^3$. On the other hand, the rubber constraint is 
\begin{equation}
\label{eq:Rubber-const-3D}
\Omega_3=0.
\end{equation}
Consider the motion of the body in the absence of potential forces, so the Lagrangian is given by the kinetic energy. Let $\| \cdot \|$ denote the euclidean norm in 
$\R^3$. Considering that $\| \dot x\|^2 =\|g^{-1}\dot x\|^2$ and
\begin{equation*}
g^{-1}\dot x = \dot X+ \Omega \times X,
\end{equation*}
we have
\begin{equation}
\label{eq:Lag-3D}
L:TQ\to \R, \qquad L(X,g,\dot X, \dot g)=\frac{1}{2}( \I \Omega, \Omega )  + \frac{m}{2} \left \| \dot X+ \Omega \times X \right \|^2,
\end{equation}
where $m$ is the total mass of the body, $\I$ is the inertia tensor and $(\cdot , \cdot )$ is the euclidean 
scalar product in $\R^3$. In the above equation it is understood that $X_3=R+a$ and $\dot X_3=0$.

There is a freedom in the choice  of orientation of the space frame which is represented by the action of $\SO(3)$ on $Q$ 
given by $h\cdot((X_1,X_2),g)=((X_1,X_2),hg)$. It is easily verified that this action satisfies the conditions (i)--(iii),
given in the definition of  a Chaplygin system in Section~\ref{S:Prelim}. Therefore, our problem is an $\SO(3)$-Chaplygin system 
with shape space $S=\R^2$.

It was shown by Borisov et al~\cite{BorMamRubber} that the system possesses a smooth invariant measure if and only if  the mass distribution of the body
is such that at least one of the following conditions is satisfied:
 \begin{itemize}
\item[C1:] The inertia tensor $\I=\mbox{diag}(I_1,I_2,I_3)$ and $a=0$.
\item[C2:] The inertia tensor $\I=\mbox{diag}(I_1,I_1,I_3)$.
\end{itemize}
Since the shape space $S=\R^2$ is two-dimensional, the Hamiltonisation of the system if either C1 or C2 holds follows from Chaplygin's
Reducing Multiplier Theorem~\ref{th:Chaplygin}.

\begin{remark}
\label{Rmk:C1C2}
Note that the inertia tensor in C1 is non-generic since  an assumption on the orientation of the body frame has already been made.
The condition C2 says that the body is axially symmetric about an axis perpendicular to the flat face.
\end{remark}

\subsection{The $n$-dimensional case}

The multi-dimensional generalisation of the system treated in the previous section consists of an $n$-dimensional rigid body
with a flat $n-1$-dimensional face that rolls without slipping or spinning about a fixed $n-1$-dimensional sphere of 
radius $R$ centred at $O$ in $\R^n$. We follow the notation of the previous section. Most of the formulas  admit a straightforward  generalisation.

As before, assume that the  body frame $\{E_1, \dots , E_n\}$ is attached to the centre of mass $C$ of the body and 
$E_n$ is parallel to the outward normal vector to the sphere at the contact point. The space frame $\{e_1, \dots , e_n\}$ has its origin at $O$.

 The $n^{th}$ entry of the 
vector $X\in \R^n$, that gives body coordinates of the vector $\overrightarrow{OC}$,  satisfies the  holonomic constraint
\begin{equation}
\label{eq:hol-const-nD}
X_n=R+a,
\end{equation}
that generalises \eqref{eq:hol-const}. The configuration space of the system is $$Q= \R^{n-1} \times \SO(n) \ni ((X_1, \dots, X_{n-1}),g),$$
 where
the attitude matrix $g$ is the change of basis matrix between the body and the space frame.
The angular velocity in the body frame is the skew symmetric matrix
\begin{equation*}
\Omega=g^{-1}\dot g\in \so(n),
\end{equation*}
with entries $\Omega_{\mu \nu}$, $1\leq \mu, \nu \leq n$.
The rolling constraints \eqref{eq:Rolling-const-3D} generalise to 
\begin{equation}
\label{eq:Rolling-const-nD}
\dot X=  -R \Omega E_n,
\end{equation}
which in particular imply $\dot X_n=0$ in consistency with \eqref{eq:hol-const-nD}.
On the other hand, the natural generalisation of the rubber constraints \eqref{eq:Rubber-const-3D} that prohibit spinning is
\begin{equation}
\label{eq:Rubber-const-nD}
\Omega_{\mu \nu}=0, \qquad 1\leq \mu, \nu \leq n-1.
\end{equation}

Now recall that for an $n$-dimensional rigid body the inertia tensor $\I$ of the body is an operator
\begin{equation*}
\I:\so(n) \to \so (n), \qquad \I(\Omega)=\J\Omega + \Omega\J,
\end{equation*}
where $\J$ is the so-called mass tensor of the body, which is a symmetric and positive definite $n\times n$ matrix (see e.g. \cite{Ratiu80}).
 The Lagrangian is
\begin{equation}
\label{eq:Lag-nD}
L:TQ\to \R, \qquad L(X,g,\dot X, \dot g)=\frac{1}{2} ( \I \Omega, \Omega  )_\kappa  + \frac{m}{2} \left \| \dot X+ \Omega  X \right \|^2,
\end{equation}
where $\| \cdot \|$ is the euclidean norm in $\R^n$, $X=(X_1, \dots, X_{n-1}, R+a)$, $\dot X_n=0$, and $( \cdot, \cdot)_\kappa$ is the Killing metric in $\so(n)$:
\begin{equation*}
(\xi, \eta)_\kappa =-\frac{1}{2}\tr(\xi \eta).
\end{equation*}

In analogy with the 3-dimensional case, it is easy to establish that the multi-dimensional 
 problem is an $\SO(n)$-Chaplygin system with $r=n-1$-dimensional shape space $S=\R^{n-1}$.
The following theorem generalises the situation that was found in the $3$-dimensional case.

\begin{theorem}
\label{th:example}
Let $n\geq 3$. The reduced system on $T^*\R^{n-1}$ possesses an invariant measure and is Hamiltonisable if any of the following 
two conditions hold
 \begin{itemize}
\item[C1:] The mass tensor $\J=\mbox{\em diag}(J_1,J_2,\dots ,J_n)$ and $a=0$.
\item[C2:] The mass tensor $\J=\mbox{\em diag}(J_1,\dots ,J_1,J_n)$.
\end{itemize}
\end{theorem}
Note that a multi-dimensional generalisation of Remark~\ref{Rmk:C1C2} also applies. 

\begin{proof}
The proof is an application of Theorem~\ref{th:Chaplygin-Gen}. The main task is to compute the gyroscopic coefficients $C_{ij}^k$ in  the  coordinates\footnote{Throughout this section we use sub-indices instead of super-indices on the coordinates.}
\begin{equation*}
s_i = X_i, \qquad i=1, \dots, n-1,
\end{equation*}
that provide a global chart for $S=\R^{n-1}$. The proof of the following lemma is given at the end of the section.
\begin{lemma}
\label{L:example}
The gyroscopic coefficients $C_{ij}^k$ written in the coordinates $(s_1, \dots, s_{n-1})$ are given as follows in the cases
C1 and C2 described in Theorem~\ref{th:example}:
 \begin{itemize}
\item[C1:] $C_{ij}^k=0$ for all $i,j,k\in \{1, \dots, n-1\}$.
\item[C2:] 
\begin{equation}
\label{eq:C2-example}
C_{ij}^k=
\frac{-ma}{R(J_1+J_n+ma^2)} (s_i  \delta_j^k -s_j\delta_i^k ), \qquad 1\leq i,j,k \leq n-1.
\end{equation}
\end{itemize}
\end{lemma}

It follows from the lemma that if C1 holds, then  the reduced system \eqref{eq:reduced-Hamilt} is already Hamiltonian without the need of a time reparametrisation, and the
symplectic volume form on $T^*\R^{n-1}$ is preserved. A similar phenomenon is encountered in the example of a homogeneous vertical disk that rolls on the
plane~(see e.g. \cite{BlochBook,LGN-Marrero}).

On the other hand, if C2 holds, then \eqref{eq:C2-example} implies that
 \ref{hyp} is verified and~\eqref{eq:H-alt} holds with $\chi_i=\frac{-ma}{R(J_1+J_n+ma^2)} s_i $.
Moreover, the conditions~\eqref{eq:Inv-Measure} for the preservation of the measure $\mu=\exp(\sigma(s))\,\nu$   are satisfied with
\begin{equation*}
\sigma =\frac{(n-2)ma}{2R(J_1+J_n+ma^2)}  \sum_{i=1}^{n-1} s_i^2.
\end{equation*}
Therefore, Theorem~\ref{th:Chaplygin-Gen} applies and the system is also Hamiltonisable in this case.
\end{proof}

The details about the Hamiltonisation stated in Theorem~\ref{th:example} are given in the following corollary that is a direct consequence of the
proof given above.

\begin{corollary}
If the condition  C1 in Theorem~\ref{th:example} holds, then the reduced equations of motion on $T^*\R^{n-1}$ are Hamiltonian
in the natural time variable and preserve the Liouville measure $\nu$ in $T^*\R^{n-1}$.

If the condition  C2 in Theorem~\ref{th:example} holds, then the reduced equations of motion on $T^*\R^{n-1}$ preserve the measure
\begin{equation*}
\mu= \exp \left ( \frac{(n-2)ma}{2R(J_1+J_n+ma^2)}  \sum_{i=1}^{n-1} X_i^2 \right ) \, \nu,
\end{equation*}
and become Hamiltonian after the time and momentum reparametrisation
\begin{equation*}
\begin{split}
dt &= \exp \left ( \frac{-ma}{2R(J_1+J_n+ma^2)}  \sum_{i=1}^{n-1} X_i^2 \right ) \, d\tau, \\
p_i & =\exp \left ( \frac{-ma}{2R(J_1+J_n+ma^2)}  \sum_{i=1}^{n-1} X_i^2 \right ) \, \tilde p_i, \qquad i=1, \dots, n-1.
\end{split}
\end{equation*}
\end{corollary}

We finally present:
\begin{proof}[Proof of Lemma~\ref{L:example}]
Throughout the proof, for $u,v\in \R^n$ we denote
\begin{equation*}
u\wedge v= uv^T-vu^T\in \so(n).
\end{equation*}
The constraints \eqref{eq:Rolling-const-nD} and \eqref{eq:Rubber-const-nD}  imply 
\begin{equation}
\label{eq:aux0}
\Omega=-\frac{1}{R}\dot X \wedge E_n.
\end{equation}

In what follows, we  abbreviate $s=(s_1, \dots,s_{n-1})\in \R^{n-1}$ and $\tilde s=(s,R+a)\in \R^n$.

 For $q=(s,g)\in Q$, identify $T_qQ=\R^n \times \so(n)$ using the  the left trivialisation of $T_g\SO(n)$ and the standard
identifications and imbedding $T_s\R^{n-1}=\R^{n-1}= \R^{n-1}\times \{0\} \hookrightarrow \R^n$.

Equation \eqref{eq:aux0} implies that:
\begin{equation*}
 \mbox{  hor}_q \left ( \frac{\partial}{\partial s_i}  \right ) = \left ( E_i, -\frac{1}{R} E_i\wedge E_n \right ), \qquad
 i=1,\dots, n-1.
\end{equation*}

 As before, denote by $\langle \cdot,
\cdot \rangle$ the Riemannian metric on $Q$ defined by the  kinetic energy Lagrangian \eqref{eq:Lag-nD}. We have
\begin{equation*}
\begin{split}
K_{kl}(s)&=\left \langle  \mbox{  hor}_q \left ( \frac{\partial}{\partial s_k}  \right ),  \mbox{  hor}_q \left ( \frac{\partial}{\partial s_l}  \right )
\right \rangle_q \\
&= \frac{1}{R^2}( \I(E_k\wedge E_n) , E_l\wedge E_n )_\kappa +m\left (E_k  -\frac{1}{R} (E_k\wedge E_n) \tilde s, E_l -\frac{1}{R} (E_l\wedge E_n) \tilde s\right )_{\R^n},
\end{split}
\end{equation*}
where $(\cdot, \cdot)_{\R^n}$ denotes the euclidean inner product in $\R^n$.
Performing the calculations the above expression simplifies to 
\begin{equation}
\label{eq:Metric-Example}
K_{kl}(s)=\frac{1}{R^2} \left ( \J_{kl} +ms_ks_l +(\J_{nn}+ma^2)\delta_{kl} \right ), \qquad 1\leq k,l\leq n-1.
\end{equation}

On the other hand the commutator
\begin{equation*}
\left [  \mbox{  hor}_q \left ( \frac{\partial}{\partial s_i}  \right ),  \mbox{  hor}_q \left ( \frac{\partial}{\partial s_j}  \right ) \right ]
=\frac{1}{R^2} \left (0, [ E_i\wedge E_n, E_j\wedge E_n]_{\so(n)}  \right )=- \frac{1}{R^2} \left (0,  E_i\wedge  E_j  \right ),
\quad 1\leq i,j \leq n-1, 
\end{equation*}
where $[\cdot , \cdot ]_{\so(n)}$ is the Lie algebra commutator in $\so(n)$. Whence,
\begin{equation*}
\begin{split}
& \left \langle
 \left [  \mbox{  hor}_q \left ( \frac{\partial}{\partial s_i}  \right ) ,  \mbox{  hor}_q \left ( \frac{\partial}{\partial s_j}  \right ) \right ] \,
, \,  \mbox{  hor}_q \left ( \frac{\partial}{\partial s_l}  \right ) 
 \right \rangle_q
  \\
&\qquad \qquad \qquad = \frac{1}{R^3}( \I(E_i\wedge E_j) , E_l \wedge E_n )_\kappa +m\left (  -\frac{1}{R^2} (E_i\wedge E_j) \tilde s, E_l -\frac{1}{R} (E_l\wedge E_n) \tilde s\right )_{\R^n}.
\end{split}
\end{equation*}
Under the assumption that $\J$ is diagonal, the above expression simplifies to
\begin{equation}
\label{eq:triple-product-example}
\left \langle \left [  \mbox{  hor}_q \left ( \frac{\partial}{\partial s_i}  \right ) ,  \mbox{  hor}_q \left ( \frac{\partial}{\partial s_j}  \right ) \right ] \,
, \,  \mbox{  hor}_q \left ( \frac{\partial}{\partial s_l}  \right ) 
 \right \rangle_q
= \frac{ma}{R^3} (s_j \delta_{il} -s_i \delta_{jl}), \qquad 1\leq i,j,l \leq n-1.
\end{equation}

If C1 holds then $a=0$ and 
all of the gyroscopic coefficients  $C_{ij}^k$ vanish in view of \eqref{eq:Coeffs}.

It remains to consider the case C2. Assume thus that $\J=(J_1,\dots, J_1, J_n)$. 
Then \eqref{eq:Metric-Example} simplifies to 
\begin{equation*}
K_{kl}(s)=\frac{1}{R^2} \left ( (J_1+J_n+ma^2)\delta_{kl} +ms_ks_l  \right ), \qquad 1\leq k,l\leq n-1,
\end{equation*}
and hence, the coefficients $C_{ij}^k(s)$ defined by \eqref{eq:C2-example} satisfy
\begin{equation*}
\begin{split}
\sum_{k=1}^{n-1} K_{kl}(s) C_{ij}^k(s) & = \frac{ma}{R(J_1+J_n+ma^2)}( -K_{jl}  s_i   +K_{il}s_j) \\
&=\frac{ma}{R^3} \left ( -s_i \delta_{jl} +s_j \delta_{il} \right ) \\
&=\left \langle \left [  \mbox{  hor}_q \left ( \frac{\partial}{\partial s_i}  \right ) , 
 \mbox{  hor}_q \left ( \frac{\partial}{\partial s_j}  \right ) \right ] \,, \,  \mbox{  hor}_q \left ( \frac{\partial}{\partial s_l}  \right ) 
 \right \rangle_q,
\end{split}
\end{equation*}
where the last identity follows from \eqref{eq:triple-product-example}. This completes the proof since the above
equation is equivalent to \eqref{eq:Coeffs}.
\end{proof}

\vspace{1cm}

\noindent \textbf{Acknowledgements:} I acknowledge the Alexander von Humboldt Foundation  for a  Georg Forster Experienced Researcher Fellowship   
  that funded a research visit to TU Berlin where this work was done.

I am extremely grateful to Y. Suris for numerous conversations during my research stay at TU Berlin.
 It was him who, in the course of these conversations, recognised the potential relevance of hypothesis \ref{hyp} and encouraged me
to investigate it further.

I also thank  J. Koiller and Y. Fedorov for their comments on an early version of this paper and 
J.C. Marrero for  conversations that were useful to understand the intrinsic
character of the gyroscopic coefficients. I thank C. Fern\'andez for her help to produce Figure~\ref{F:fig1}.

Finally, I express my gratitude to the anonymous referees of the paper who helped me improve the exposition with their
useful comments and remarks.

\vskip 1cm

\noindent
 LGN: Departamento de Matem\'aticas y Mec\'anica, IIMAS-UNAM.
Apdo. Postal 20-126, Col. San \'Angel,
Mexico City, 01000, Mexico. luis@mym.iimas.unam.mx

\end{document}